\documentclass[letterpaper, 10 pt, conference]{ieeeconf}  %

\IEEEoverridecommandlockouts                              %
\usepackage{caption}
\usepackage{amsthm}
\usepackage{amsmath}
\usepackage{amssymb}
\usepackage{mathtools}
\usepackage{hyperref}
\usepackage{cleveref}
\usepackage{bm}
\usepackage{todonotes}
\usepackage{cite}

\makeatletter
\def\namedlabel#1#2{\begingroup
    #2%
    \def\@currentlabel{#2}%
    \phantomsection\label{#1}\endgroup
}
\makeatother

\newtheorem{definition}{Definition}
\newtheorem{theorem}{Theorem}
\newtheorem{lemma}{Lemma}
\newtheorem{corollary}{Corollary}
\newtheorem{remark}{Remark}
\newtheorem{setting}{Setting}
\newtheorem{definition*}{Definition}
\newtheorem{example}{Example}
\newtheorem{proposition}{Proposition}
\newtheorem*{proposition*}{Proposition}
\newtheorem*{theorem*}{Theorem}
\newtheorem*{lemma*}{Lemma}
\newtheorem*{corollary*}{Corollary}
\newtheorem*{remark*}{Remark}
\newtheorem*{setting*}{Setting}
\newtheorem*{example*}{Example}

\newcommand{\Nc}{\mathcal N}

\newcommand{\N}{\mathbb N}
\newcommand{\E}{\mathbb E}
\newcommand{\R}{\mathbb R}

\newcommand{\Var}{\mathrm{Var}}

\newcommand{\Fc}{\mathcal F}

\newcommand{\sign}{\mathrm{sign}}

\title{\LARGE \bf
   On the continuity and smoothness of the value function in reinforcement learning and optimal control
}

\author{Hans Harder$^1$ %
and Sebastian Peitz$^1$%
\thanks{
   $^1$Department of Computer Science, Paderborn University, Paderborn, Germany.
   { \tt\{hans.harder,sebastian.peitz\}@upb.de}.}%
}

\begin{document}

\maketitle
\thispagestyle{empty}
\pagestyle{empty}

\begin{abstract}

   The \emph{value function} plays a crucial role as a measure for the cumulative future reward an agent receives in both reinforcement learning and optimal control. It is therefore of interest to study how similar the values of neighboring states are, i.e., to investigate the continuity of the value function. We do so by providing and verifying upper bounds on the value function's modulus of continuity. Additionally, we show that the value function is always Hölder continuous under relatively weak assumptions on the underlying system and that non-differentiable value functions can be made differentiable by slightly ``disturbing'' the system.

\end{abstract}

\section{Introduction}

In reinforcement learning, an agent is put into some environment, acts in it, accumulates traces of rewards, and tries to find actions that will give it the best possible outcome. This is achieved by learning a policy (or feedback law) that maximizes the cumulative future reward it receives, measured by the value function.

A well-known two-phase approach to this issue is \emph{policy iteration} (see \cite{suttonReinforcementLearningIntroduction1998}): In the \emph{policy evaluation} phase, the agent approximates the value (i.e., cumulative reward) for each state by collecting trajectories following the current policy. %
In the \emph{policy improvement} phase, the agent adapts its policy by acting greedily with respect to the value function. %

We study properties of the value function for a fixed policy, as in the policy evaluation phase. (In optimal control, it is convention to only call the cost functional evaluated for the optimal policy the ``value function''. In our case, every policy corresponds to a value function.)
Essentially, this allows us to discard the policy in our analysis by viewing it as part of the environment. 
This setup is common in the analysis of temporal-difference learning, see, e.g., \cite{tsitsiklisAnalysisTemporaldifferenceLearning1997,suttonFastGradientdescentMethods2009,boyanTechnicalUpdateLeastSquares2002}. 

For illustration purposes, let us introduce the formulation for the case where both environment and policy are deterministic. We are given a set of states $S$ and actions $A$, a model of the environment $\Phi : S \times A \rightarrow S$, a reward function $r : S \times A \rightarrow \R$ and a discount factor $\gamma \in (0,1)$. 
When a policy $\pi : S \rightarrow A$ is fixed, the model reduces to a dynamical system $\Phi_\pi : S \rightarrow S$ defined by $\Phi_\pi(x) = \Phi(x,\pi(x))$, and the reward function is simplified by introducing $r_\pi(x) = r(x,\pi(x))$. The \emph{value function} is
\begin{align}
   \label{eq:value-function-with-policy}
   v_\pi(x) = \sum_{n=0}^\infty \gamma^n r_\pi( \Phi_\pi^n(x)),
\end{align}
where $(\cdot)^n$ denotes the $n$-fold application of $\Phi_\pi$. 
Accordingly, one can define the \emph{state-action value} as 
\begin{align}
   \label{eq:state-action-value-function-with-policy}
   q_\pi(x,a) = r(x,a) + \gamma v_\pi(\Phi(x,a)).
\end{align}
\Cref{eq:value-function-with-policy} is more or less the form that we are interested in, and we will discard the subscript $(\cdot)_\pi$ in the forthcoming investigations by assuming that $\pi$ is fixed. In other words, we will think of $r$ and $\Phi$ as functions from $S$ to $\R$ and $S$ respectively, independent of any action.

Works such as \cite{silverDeterministicPolicyGradient2014} or \cite{lillicrapContinuousControlDeep2016a} approximate \eqref{eq:state-action-value-function-with-policy} using neural networks and rely on taking derivatives with respect to $a$. In practice, this is not a problem since these models are differentiable by definition. However, when it comes to the \emph{true} state-action value function, differentiability of $v_\pi$ is required. But this is an assumption that needs good justification; there are value functions that are \emph{nowhere} differentiable:

\begin{proposition}[see \cite{yamagutiWeierstrassFunctionChaos1983}]
   \label{thr:logistic}
   Let $\Phi(x) = 4x(1-x)$ be the logistic map on $S=[0,1]$, put $r(x) = x$ and let $\gamma \in [\frac 1 2, 1)$. Then $v(x)=\sum_{n=0}^\infty \gamma^n r(\Phi^n(x))$ is nowhere differentiable.
\end{proposition}

\begin{remark}
   One can obtain the system of \Cref{thr:logistic} by defining $r(x,a) = x$ and $\Phi(x,a) = a \cdot 4x(1-x)$ for $A = S = [0,1]$. The policy $\pi(x) = 1$ induces the logistic map. %
\end{remark}

On the other hand, we \emph{can} get differentiability by ``disturbing'' the system. Informally, if we define a new random dynamical system via ``original system + noise'', then the resulting value function is differentiable. We have illustrated this in \Cref{fig:example-logistic} using the logistic map from \Cref{thr:logistic}. %

\begin{figure}
   \centering
   \includegraphics[width=0.31\textwidth]{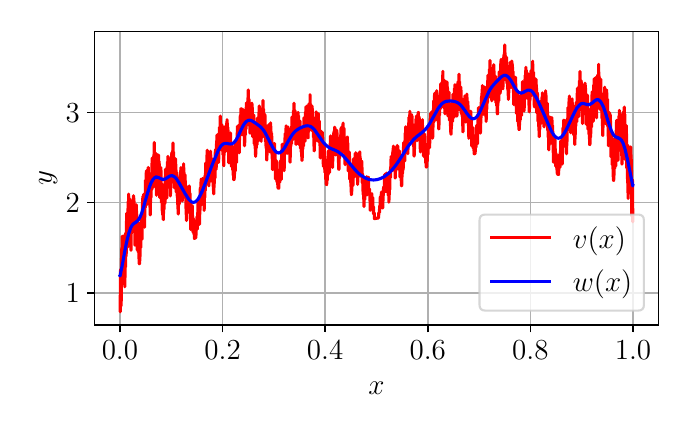}
   \vspace{-10pt}
   \caption{The value function $v$ from \Cref{thr:logistic} for the discount factor $\gamma = 0.8$. The ``smoothed'' version $w$ is the value function that one obtains when disturbing the same system using Gaussian noise with standard deviation $\sigma = 0.01$, cf. \Cref{thr:differentiability}. }
   
   \label{fig:example-logistic}
   \vspace{-16pt}
\end{figure}

In this paper, we will study the general case, i.e., we also analyze nowhere differentiable value functions. 
It turns out that these functions are Hölder continuous under weak assumptions, which is useful since this can be used to bound the variance of $v$ if evaluated at random states:

\begin{proposition}
   \label{thr:varbound}
   Suppose $v : S \rightarrow \R$ for $S \subseteq \R^N$ is Hölder continuous, i.e., there are $A \geq 0$ and $0 < \beta \leq 1$ such that
   $
      |v(x)-v(x')| \leq A || x - x'||_2^\beta
   $
   for all $x,x' \in S$. Let $X$ be an $S$-valued random variable. Then 
   \begin{equation}
      \label{eq:variance-decomposition}
      \Var[v(X)] \leq \frac {A^2} {2^{1-\beta}} ( \sum_{i=1}^N \Var[X_i] )^\beta.
   \end{equation}
\end{proposition}
\begin{proof}
   Any real-valued random variable $Y$ satisfies $\Var[Y] = \frac 1 2 \E[(Y-Y')^2]$, where $Y'$ is an independent copy of $Y$.
   Hence, by Jensen's inequality:
   \begin{align*}
      \Var[v(X)] 
      &= \textstyle\frac 1 2 \E[(v(X)-v(X'))^2] \\
      &\leq \textstyle\frac 1 2 A^2 \E[|| X - X'||^{2\beta}_2] \\
      &\leq \textstyle\frac 1 2 A^2 \E[|| X - X'||^2_2]^\beta \\
      &= \textstyle\frac 1 {2^{1-\beta}} A^2 (\sum_{i=1}^N \frac 1 2 \E[(X_i-X_i')^2])^\beta,
   \end{align*}
   which is just \eqref{eq:variance-decomposition}.
\end{proof}

We study the continuity of $v$ by categorizing it as ``Lipschitz continuous'', ``Hölder continuous'' or as a mix thereof, depending on the relationship between the Lipschitz constant of the system and the discount factor. Hölder continuity of $v$ is known under weak assumptions for the deterministic and continuous-time setting, see \cite{bardiOptimalControlViscosity1997}, and it is known for the discrete-time setting under similar assumptions as we do, see \cite{bernsteinAdaptiveresolutionReinforcementLearning2010}. We improve these results by giving a closed-form solution to an upper bound of $v$'s modulus of continuity and show that it cannot be improved for linear reward functions. Additionally, we consider both random and deterministic systems, as well as continuous-time and discrete-time systems. Finally, we show how ``disturbed'' systems can yield differentiable value functions.
Besides \cite{bardiOptimalControlViscosity1997} and \cite{bernsteinAdaptiveresolutionReinforcementLearning2010}, there have been investigations under varying assumptions by \cite{struloviciSmoothnessValueFunctions2015a,rincon-zapateroDifferentiabilityValueFunction2009,itoContinuityValueFunction2020a,bascoLipschitzContinuityValue2019} and \cite{huangUniquenessConstrainedViscosity2007}.
\Cref{thr:logistic} is a result of \cite{yamagutiWeierstrassFunctionChaos1983}, who studied the equation in a seemingly unrelated setting.

\subsubsection*{Organization}

In \Cref{sec:continuity}, we derive continuity bounds for the value function in both continuous- and discrete-time settings. We show how these bounds improve on and entail known results about the Lipschitz resp. Hölder continuity of the value function. The derived upper bound from \Cref{sec:continuity} is verified experimentally and shown to be sharp in \Cref{sec:sharpness_and_example}. Finally, we show how to obtain differentiability as soon as we add a little bit of Gaussian noise in \Cref{sec:differentiability}.

\section{Definitions}

We fix a probability triple $(\Omega, \Fc, P)$ and a metric space $(S,d)$. We call $S$ the \emph{state space} and suppose that $d$ is bounded; $\sup_{x,x'\in S} d(x,x') \leq D$. For technical reasons, we also have to assume that $S$ is part of a measurable space $(S,\Sigma)$ so that measurable functions from $S$ to $\R$ can be composed or integrated. 
Forthcoming, we will frequently use the \emph{modulus (of continuity)} of a uniformly continuous function $f : S \rightarrow \R$. This is the function defined on $\R_{\geq 0}$ by
\begin{align*}
   W_f(d_0) = \sup_{x,x'\in S} \{ |f(x)-f(x')| : d(x,x') \leq d_0 \}.
\end{align*}
Intuitively, $W_f$ is the best function that one can use when bounding $|f(x)-f(x')|$ in terms of $d(x,x')$.

Next, we define ``random systems''. Such systems model random transitions between states over given periods. We also introduce the notion of \emph{Lipschitz continuity in Expectation}, or LE-continuity, which is used to bound the divergence between nearby states:

\begin{definition}
   \label{def:simplified-stochastic-flow}
   Let $T \in \{ \R_{\geq 0}, \N_0\}$ be the time domain. A measurable function $\Phi : T \times S \times \Omega \rightarrow S$, written $\Phi_t(x) = \Phi_t(x,\cdot) = \Phi(t,x,\cdot)$, is called a random system (so that $\Phi_t(x)$ is a random variable). If $T = \R_{\geq 0}$ ($T = \N_0$), we call $\Phi$ a continuous-time (discrete-time) system.
   We say that $\Phi$ is LE-continuous with constant $L \geq 0$ when
   \begin{align*}
      \E[ d(\Phi_t(x),\Phi_t(x')) ] \leq L^t d(x,x') \quad \forall t \in T,\, x,x' \in X. 
   \end{align*}
\end{definition}

This definition is very general and not strong enough to obtain the well-known recursive representation of the value function. For this one needs stochastic flows \cite{dorogovtsevAnalysisStochasticFlows2014, dorogovtsevMeasurevaluedProcessesStochastic2023}, random dynamical systems \cite{arnoldRandomDynamicalSystems1995} or Markov chains, but the latter does not allow for a simple characterization of LE-continuity (see \cite{dorogovtsevMeasurevaluedProcessesStochastic2023}). However, the required properties from \Cref{def:simplified-stochastic-flow} already suffice to get all desired results since we study value functions in their infinite-sum (or infinite-integral) form, which does not require the Markov property.

\begin{example}
   Let $\xi_0,\xi_1,\dots : \Omega \rightarrow \R^n$ be random vectors and suppose $L \geq 0$ and $\varphi: S \times \R^n \rightarrow S$ satisfy
   \begin{align*}
      d(\varphi(x,y), \varphi(x',y)) \leq L d(x,x'), \quad \forall x,x' \in S, y \in \R^n.
   \end{align*}
   Then the system defined by $\Phi_{n+1}(x) = \varphi(\Phi_n(x),\xi_n)$ and $\Phi_0(x) = x$ 
   is LE-continuous with constant $L$.
\end{example}

\begin{remark}
   If the discrete metric $d(x,x')=1_{x\neq x'}$ is used, e.g., if $S$ is countable, then LE-continuity is trivially satisfied for $L= 1$. LE-continuity is therefore only necessary when $S$ is uncountable, e.g. because it is an open or closed subset of $\R^N$. Intuitively, it requires that trajectories cannot diverge too fast if the distance of the initial states is small.
\end{remark}

\begin{remark}
   The at worst exponential divergence of trajectories is a common assumption, see, e.g., \cite{bardiOptimalControlViscosity1997}. We alter it slightly by introducing the expectation. This does not make a difference for deterministic systems.
\end{remark}

If a random system is equipped with a reward function and a discount factor, one obtains a corresponding value function. We discard the policy and the set of actions:

\begin{setting}
   \label{def:rl-setting-lipschitz}
   Let $r : S \rightarrow \R$ be a reward function for which there is a concave and monotonically increasing $R : \R \rightarrow \R$ with $R(0)=0$ such that $|r(x)-r(x')| \leq R(d(x,x'))$ for all $x,x' \in S$. Moreover, let $\Phi$ be a random system that is LE-continuous with constant $L > 0$. Let $\gamma \in (0,1)$ be a discount factor. Define the value function as
   \begin{align*}
      v(x) = \begin{cases}
         \int_0^\infty \gamma^t \E[r(\Phi_t(x))] dt & T = \R_{\geq 0}, \\
         \sum_{n=0}^\infty \gamma^n \E[r(\Phi_n(x))] & T = \N_0.
      \end{cases}
   \end{align*}
\end{setting}

\section{Continuity of the value function} \label{sec:continuity}

Assuming that \Cref{def:rl-setting-lipschitz} holds, the overall goal of this section is to bound the difference $|v(x)-v(x')|$ by a quantity $H(d(x,x'))$, where $H$ is a function that is to be derived. The derived $H$ depends on multiple factors: the Lipschitz constant $L$ and the bound $R$, the discount factor $\gamma$, the maximal distance between two states $D$, and finally the chosen time domain. In \Cref{thr:hoelder-integrals}, we derive $H$ for the continuous case since the discrete case can be reduced to it. We assume $L>1$ since $v$ is trivially Lipschitz continuous otherwise (see \Cref{sec:hoelder-and-lipschitz} for the case $L\leq 1<1/\gamma$).

\begin{figure}
   \centering
   \includegraphics[width=0.3\textwidth]{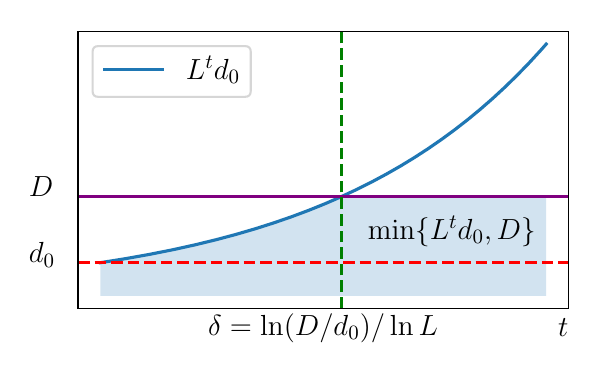}
   \vspace{-10pt}
   \caption{Visual depiction of the idea in \Cref{thr:hoelder-integrals}.}
   \label{fig:proof}
   \vspace{-15pt}
\end{figure}

\begin{theorem}
   \label{thr:hoelder-integrals}
   Suppose \Cref{def:rl-setting-lipschitz} holds with $T = \R_{\geq 0}$ and $L > 1$. Define for $d_0 > 0$:
   \begin{align*}
      K(d_0) = 
      \begin{dcases}
         \frac{\ln L}{\ln (\gamma L)} 
            D^{1-\beta} d_0^\beta + \frac{\ln \gamma}{\ln (\gamma L)} d_0 & L \neq 1/\gamma, \\
         ( \ln(D/d_0) + 1) d_0 & L = 1/\gamma,
      \end{dcases}
   \end{align*}
   with $K(0) := 0$ and $\beta = \ln (1/\gamma) / \ln L$. Setting 
   \begin{align}
      H(d_0) = \frac{R(K(d_0))}{\ln(1/\gamma)},
   \end{align}
   we have
   \begin{align}
      |v(x)-v(x')| \leq 
         H(d(x,x')),
      \quad \forall x,x' \in S.
      \label{eq:hoelder-continuity}
   \end{align}
\end{theorem}
\begin{proof}
   Take $x,x' \in S$ and define $d_0 := d(x,x')$. For $d_0 = 0$, use that $d$ is a metric to obtain $x = x'$ and thus $|v(x)-v(x')|=0$. We have $K(0) = 0$ and $R(0) = 0$ by definition, hence \eqref{eq:hoelder-continuity} applies. Forthcoming, suppose $d_0 > 0$.
   Using the concavity of $R$ with Jensen's inequality, we get
   \begin{align}
      \label{eq:value-bound-pf-1}
      |v(x)-v(x')| \leq \int_0^\infty \gamma^t R(\E[d(\Phi_t(x), \Phi_t(x'))]) dt.
   \end{align}
   We introduce the factor $\ln(1/\gamma) / \ln(1/\gamma)$ and notice that $\int_0^\infty \ln(1/\gamma) \gamma^t dt =1$. Thus, we may apply Jensen's inequality a second time to bound $|v(x)-v(x')|$ by:
   \begin{align}
      \label{eq:integral_bound_raw}
      \frac{R\big(\ln(1/\gamma)\int_0^\infty \gamma^t \E[d(\Phi_t(x), \Phi_t(x'))] dt\big)}{\ln(1/\gamma)}.
   \end{align}
   Introducing
   \begin{align}
      \label{eq:K_delta}
      K_\delta(d_0) = \ln(1/\gamma) \big( \int_0^\delta (\gamma L)^t d_0 dt + \int_\delta^\infty \gamma^t D dt \big),
   \end{align}
   we find that for every $\delta \geq 0$:
   \begin{align}
      \label{eq:K_delta_bound}
      \ln(1/\gamma) \int_0^\infty \gamma^t \E[d(\Phi_t(x),\Phi_t(x'))] dt \leq K_\delta(d_0)
   \end{align}
   by bounding $\E[d(\Phi_t(x),\Phi_t(x'))]$ by $L^t d_0$ in the interval from $0$ to $\delta$ and by $D$ in the interval from $\delta$ to $\infty$.

   Since \eqref{eq:K_delta} is a bound for arbitrary $\delta$, we can find a $\delta$ that is optimal, i.e., minimizes $K_\delta(d_0)$: The term $L^t d_0$, as a function of $t$, is strongly monotonically increasing if $L > 1$, and thus it must eventually exceed $D$. It is thus advantageous to bound $\E[d(\Phi_t(x),\Phi_t(x'))]$ by $D$ after this ``crossing'' occurs and better to bound it by $L^t d_0$ before, and hence the optimal choice for $\delta$ is
   $$
      L^\delta d_0  = D \implies \delta = \ln( D / d_0) / \ln L.
   $$
   \Cref{fig:proof} shows the intuitive idea of this choice.
   The quantity $\delta$ must be finite and non-negative since $D > d_0 > 0$ and $L > 1$.
   Indeed, for this particular choice,
   \begin{align*}
      L^\delta = e^{\ln (L) \delta} = D/d_0
      \quad \text{and}\quad
      \gamma^\delta = e^{\ln (\gamma) \delta} = D^{-\beta} d_0^\beta.
   \end{align*}
   Solving the integral for the case $L \neq 1/\gamma$:
   \begin{align*}
      K_\delta(d_0) 
      &= \ln(1/\gamma)\frac{d_0((\gamma L)^\delta-1)}{\ln(\gamma L)} +  D\gamma^\delta \\
      &= \ln(1/\gamma)\frac{d_0( D^{1-\beta} d_0^{\beta-1}-1)}{\ln(\gamma L)} + D^{1-\beta} d_0^\beta \\
      &= D^{1-\beta} d_0^\beta \big( 1 - \frac{\ln \gamma}{\ln (\gamma L)}\big) + \frac{\ln \gamma}{\ln(\gamma L)} d_0 \\
      &= \frac{\ln L}{\ln (\gamma L)} D^{1-\beta} d_0^\beta + \frac{\ln \gamma}{\ln(\gamma L)} d_0 = K(d_0),
   \end{align*}
   where we have used $1-\ln \gamma/\ln(\gamma L) = \ln L / \ln(\gamma L)$.
   In the case $L = 1/\gamma$, the integral becomes:
   \begin{align*}
      K_\delta(d_0) 
      &= \ln(1/\gamma) d_0 \delta + D \gamma^\delta  \\
      &= ( \ln(D/d_0) + 1) d_0 = K(d_0),
   \end{align*}
   using $\beta = 1$. Substitute $K_\delta(d_0) = K(d_0)$ in \eqref{eq:K_delta_bound} and apply monotonicity of $R$ in \eqref{eq:integral_bound_raw} to obtain \eqref{eq:hoelder-continuity}.
\end{proof}

\begin{theorem}
   \label{thr:hoelder-sums}
   Suppose \Cref{def:rl-setting-lipschitz} holds with $T = \N_0$ and $L > 1$, and let $H$ be as in \Cref{thr:hoelder-integrals}. Then
   \begin{align*}
      |v(x)-v(x')| \leq \frac{\ln(1/\gamma)}{1-\gamma} H(d(x,x')), \quad \forall x,x' \in X.
   \end{align*}
\end{theorem}
\begin{proof}
   We reduce the problem to a continuous setting.
   Extend $\Phi$ from $\N_0$ to $\R_{\geq 0}$ by defining
   $
      \Phi_{t}^c(x) := \Phi_{\lfloor t \rfloor}(x)
   $
   for $t \in \R_{\geq 0}$, and suppose that the value function for $\Phi^c$ is given by $v^c$. Now, for any $n \in \N_0$, 
   \begin{align*}
      \int_n^{n+1} \gamma^t \E[r(\Phi_t^c(x))] dt 
      &= \E[r(\Phi_n(x))] \int_n^{n+1} \gamma^t dt \\
      &= \gamma^n \E[r(\Phi_n(x))] \frac{1-\gamma}{\ln(1/\gamma)}.
   \end{align*}
   Therefore,
   \begin{align}
      \label{eq:pf-hoelder-cont-discrete}
      v^c(x) 
      = \sum_{n=0}^\infty \int_n^{n+1} \gamma^t \E[r(\Phi_t^c(x))] dt 
      = \frac{1-\gamma}{\ln(1/\gamma)} v(x).
   \end{align}
   Moreover, $\Phi^c$ has the same LE-constant as $\Phi$: Since $L > 1$,
   \begin{align*}
      \E[d(\Phi_t^c(x),\Phi_t^c(x'))]
      &= \E[d(\Phi_{\lfloor t \rfloor}(x),\Phi_{\lfloor t \rfloor}(x'))] \\
      &\leq L^{\lfloor t \rfloor} d(x,x')
      \leq L^t d(x,x').
   \end{align*}
   Then use \eqref{eq:pf-hoelder-cont-discrete} and \Cref{thr:hoelder-integrals} to bound $|v(x)-v(x')|$ by:
   \begin{align*}
      \frac{\ln(1/\gamma)}{1-\gamma}|v^c(x) - v^c(x') |
      \leq \frac{\ln(1/\gamma)}{1-\gamma} H(d(x,x')).
   \end{align*}
\end{proof}

\subsection{Hölder and Lipschitz continuity} \label{sec:hoelder-and-lipschitz}

The results from \cite{bardiOptimalControlViscosity1997} and \cite{bernsteinAdaptiveresolutionReinforcementLearning2010} are entailed by \Cref{thr:hoelder-integrals,thr:hoelder-sums}, but we can also extend them into the stochastic setting: Under the assumptions of \Cref{def:rl-setting-lipschitz} and if the reward function is Lipschitz continuous, i.e., $R(d_0) = C d_0$ for a $C \geq 0$, then there exist $\beta \in (0,1]$ and $A \geq 0$ such that: %
\begin{align}
   \label{eq:continuity-value-function-hoelder}
   |v(x)-v(x')|\leq A d(x,x')^\beta, \quad \forall x,x' \in S.
\end{align}
In fact, if
\begin{itemize}
\setlength{\itemindent}{40pt}
   \item[$L < 1/\gamma$,] then $\beta = 1$ (\Cref{thr:continuity-value-function-case-less}),
   \item[$L = 1/\gamma$,] then any $\beta \in (0,1)$ is allowed (\Cref{thr:continuity-value-function-case-eq}),
   \item[$L > 1/\gamma$,] then $\beta = \ln(1/\gamma) / \ln L$ (\Cref{thr:continuity-value-function-case-greater}).
\end{itemize}
The first case shows Lipschitz continuity, the second case Hölder continuity for any exponent in $(0,1)$, and the last case Hölder continuity for the exponent $\ln(1/\gamma)/\ln L$.
The results follow from \Cref{thr:hoelder-integrals,thr:hoelder-sums}, except for the straight-forward case $L < 1/\gamma$. For $L = 1/\gamma$, one can take $\beta$ as close to $1$ as one desires, but the corresponding $A$ becomes arbitrarily large. For the proofs, we refer to the appendix.

\section{Sharpness and experiments}\label{sec:sharpness_and_example}

While \Cref{thr:hoelder-integrals,thr:hoelder-sums} provide upper bounds on the modulus of $v$, it is still unclear how good the bounds are. Since they only depend on $D$, $L$, $\gamma$ and $R$, our objective is now to construct a random system that is LE-continuous with constant $L$ and whose corresponding value function has the largest possible modulus. The difference between the resulting modulus and the proposed bounds then indicates how much the latter could be decreased.
We start by looking at the deterministic map $x \mapsto x^L$ defined on $[0,1]$. Afterward, we show that one can construct a random system whose value function has a modulus that exactly matches the one proposed in \Cref{thr:hoelder-integrals}.

\subsection{An example for a steep value function} \label{sec:example}

Take an arbitrary $L > 1$ and define a dynamical system $\Phi : [0,1] \rightarrow [0,1]$ via
$
   \Phi(x) = x^L.
$
For $n \in \N_0$, set
$$
   \Phi(n,x) := \Phi_n(x) := \Phi^n(x) = x^{(L^n)}.
$$ 
The system $\Phi$ has a stable fixed point at $x = 0$ and an unstable one at $x = 1$. %
We also have
$$
   \partial_x \Phi_n(x) = L^n x^{L^n-1} \leq L^n,
$$
which shows that $\Phi_n$ is Lipschitz continuous with constant $L^n$ for the Euclidean metric $d(x,x') = |x-x'|$. If interpreted as a random system, i.e., when $\Phi_n(x,\omega) := \Phi_n(x)$ for $\omega \in \Omega$, it must be LE-continuous with constant $L$.
For $r(x) = x$ and any discount factor $\gamma \in (0,1)$, the value function is
\begin{align*}
   v(x) = \sum_{n=0}^\infty \gamma^n x^{(L^n)}.
\end{align*}
It is maximal for $x=1$, where it takes the value $1/(1-\gamma)$, and minimal for $x = 0$, where it is $0$. All other $x \in (0,1)$ quickly approach $0$ as $n$ increases. %
One can see the value functions for varying discount factors on the left in \Cref{fig:example}, depicting a steep increase as $x$ approaches $1$ for the case $L \geq 1/\gamma$. If that happens, $v$'s derivative is unbounded at $x=1$.
The modulus takes the form
$
   W_v(d_0) = v(1) - v(1-d_0),
$
see \Cref{thr:moc-example} in the appendix for details.
As observable on the right of \Cref{fig:example}, the bounds provided by \Cref{thr:hoelder-sums} are close to $W_v$ for varying values of $\gamma$. 
Bounds that only assume knowledge about $L$ and $\gamma$ can hence not be much better than \Cref{thr:hoelder-sums} -- at least for linear reward functions. %

\begin{figure}
   \centering
   \includegraphics[width=0.23\textwidth]{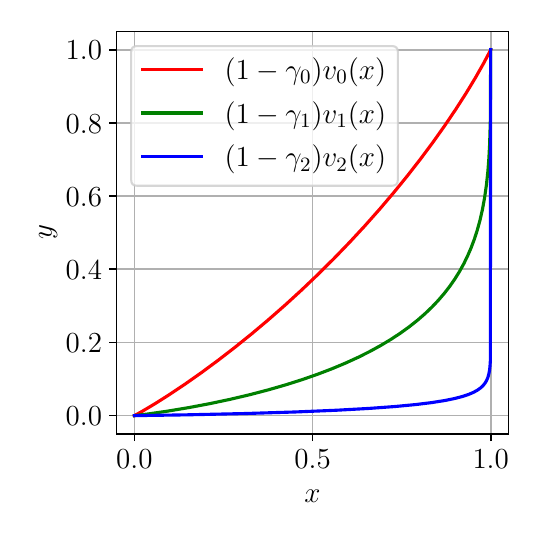}
   \includegraphics[width=0.23\textwidth]{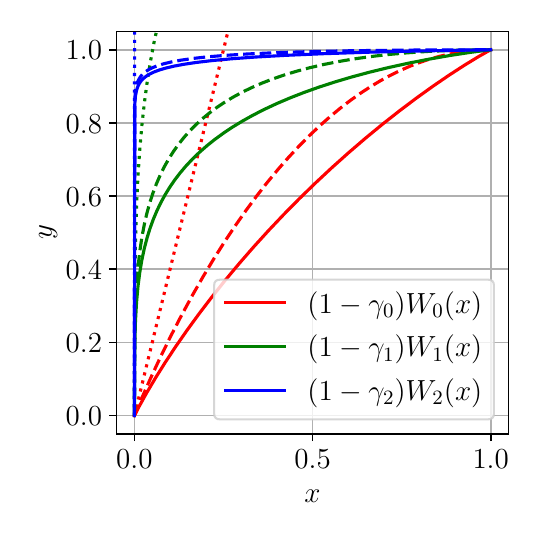}
   \caption{Left: The value functions corresponding to the example in \Cref{sec:sharpness_and_example} for $L = 1.5$ and discount factors $(\gamma_0, \gamma_1, \gamma_2) = (0.5, 0.9, 0.99)$ when normalized to a maximal value of $1$. Right: The moduli of continuity for the value functions in comparison to the bounds given by \Cref{thr:hoelder-sums}, visualized by dashed lines in the same color. The bounds from \cite{bernsteinAdaptiveresolutionReinforcementLearning2010} are visualized using dotted lines (also the same color).}
   \vspace{-1em}
   \label{fig:example}
\end{figure}

\subsection{Sharpness} \label{sec:sharpness}

In the time-continuous setting and for linear rewards, we cannot improve the bound:

\begin{proposition}
   Let $D > 0$, $L > 1$, $\gamma \in (0,1)$ and suppose $r(x)=Ax$. Then there exists a continuous-time random system with LE-constant $L$ satisfying \Cref{def:rl-setting-lipschitz} such that $v$'s modulus is exactly $H$ from \Cref{thr:hoelder-integrals}.
\end{proposition}
\begin{proof}
   Define a deterministic system on $S = [0,D]$ by $\Phi_t(x,\omega) = \min\{D, L^t x\}$ and let $d(x,x') = |x-x'|$. One may verify that $\Phi$ is indeed LE-continuous with constant $L$. 
   The value function takes the form
   \begin{align*}
      v(x) = A \big( \int^\delta_0 (\gamma L)^t x dt + \int_\delta^\infty \gamma^t D dt \big),
   \end{align*}
   where $\delta$ solves $L^\delta x = D$, i.e. $\delta = \ln(D/x)/\ln L$ for $x > 0$ and $\delta = \infty$ otherwise. This is exactly 
   $$
      v(x) = A K(x)/\ln(1/\gamma),
   $$ 
   where $K$ is from \Cref{thr:hoelder-integrals}, see \eqref{eq:K_delta} and \eqref{eq:K_delta_bound} with $K_\delta(d_0) = K(d_0)$ for $\delta = \ln(D/d_0)/\ln L$. 
   Let now $x, x' \in S$. Then
   \begin{align*}
      |v(x)-v(x')| = \frac{|A|}{\ln(1/\gamma)} |K(x)-K(x')|.
   \end{align*}
   Moreover, notice for $L \neq 1/\gamma$ and $D \geq x > 0$:
   \begin{align*}
      K'(x) &= \frac{\ln \gamma}{\ln(\gamma L)} (1-D^{1-\beta} x^{\beta-1}) \geq 0, \\
      K''(x) &= -\frac{\ln \gamma}{\ln (\gamma L)} D^{1-\beta} (\beta - 1) x^{\beta - 2} \leq 0,
   \end{align*}
   which is due to
   \begin{align*}
      \sign(\frac{\ln \gamma}{\ln(\gamma L)}) 
      = \sign(\beta - 1)
      = \sign(1-(\frac x D)^{\beta-1}),
   \end{align*} 
   for both $L < 1/\gamma$ and $L> 1/\gamma$.
   Thus, $K$ is monotonically increasing, while $K'$ is monotonically decreasing. Without loss of generality for $x \geq x'$,
   \begin{align*}
      |K(x)-K(x')| 
      &= K(x)-K(x') \\
      &= \textstyle\int_{x'}^x K'(y) dy \\
      &\leq \textstyle\int_{x'}^x K'(y-x') dy \\
      &= \textstyle\int_0^{x-x'} K'(y) dy \\
      &= K(x-x')-K(0) = K(x-x').
   \end{align*}
   Since the modulus is defined as the supremum over all $x,x'$ and since $K$ is monotonically increasing, the best way to maximize $K(x-x')$ subject to $0 \leq x-x' \leq d_0$ is to put $x = d_0$ and $x'= 0$:
   \begin{align*}
      W_v(d_0) 
      = \sup_{0 \leq x-x'\leq d_0}\frac{|A| K(x-x')}{\ln(1/\gamma)} = H(d_0),
   \end{align*}
   which completes the proof for $L \neq 1/\gamma$.
   For the case $L = 1/\gamma$, we have for $x > 0$:
   \begin{align*}
      K'(x) = \ln(D / x) \geq 0, \quad
      K''(x) = -1 / x \leq 0.
   \end{align*}
   The remainder of the argument is analogous.
\end{proof}

\section{Disturbance implies Differentiability}\label{sec:differentiability}

If one stays in the deterministic realm, then one has to live with systems that give rise to non-differentiable value functions. However: One can \emph{disturb} such systems by adding a little bit of Gaussian noise in every step. This per-step smoothing procedure suffices to obtain a differentiable value function; we illustrate this for the case $S = \R^n$: %

\begin{proposition}
   \label{thr:differentiability}
   Let $\Phi : \R^n \rightarrow \R^n$, $r : \R^n \rightarrow \R$ be a bounded and differentiable reward function, $\gamma \in (0,1)$ be a discount factor and $\sigma^2 > 0$ be a variance. Define a random system
   \begin{align*}
      \psi_{n+1}(x) = \psi_n(\Phi(x + \xi_n)), \quad \psi_0(x) = x
   \end{align*}
   with $\xi_0,\xi_1,\dots \sim \Nc(0,\sigma^2 I)$ iid and $I$ being the $n\times n$ identity matrix. Then the value function $w$ of $\psi$ 
   is differentiable, and its gradient is given for $\xi \sim \Nc(0,\sigma^2 I)$ by
   \begin{align}
      \label{eq:noise-gradient}
      \nabla w(x) = \nabla r(x) + \frac \gamma {\sigma^2} \E[w(\Phi(x+\xi))\cdot \xi].
   \end{align}
\end{proposition}
\begin{proof}
   Because of the independence of $\xi_0,\xi_1,\dots$, the new value function satisfies
   $$
      w(x) = r(x) + \gamma \E[w(\Phi(x+\xi))].
   $$
   Writing the expectation as a convolution,
   $$
      \E[w(\Phi(x+\xi))] = \E[w(\Phi(x-\xi))] = ((w\circ \Phi) * p)(x),
   $$
   where $p : \R^n\rightarrow \R$ is the density of $\Nc(0,\sigma^2 I)$. We can take the gradient with respect to $p$ \emph{instead} of $w \circ \Phi$,
   \begin{align*}
      \nabla ((w \circ \Phi) * p)(x) = ((w \circ \Phi) * \nabla p)(x).
   \end{align*}
   Using $\nabla p(\xi) = -\frac 1 {\sigma^2} p(\xi) \cdot \xi$,
   \begin{align*}
      \nabla \E[w(\Phi(x+\xi))] = \frac{1}{\sigma^2} \E[w(\Phi(x - \xi)) \cdot (-\xi)],
   \end{align*}
   which gives \eqref{eq:noise-gradient} after changing variables from $-\xi$ to $\xi$. The gradient $\nabla w(x)$ must also be finite, which follows from the fact that $\xi$ has finite variance and from the boundedness of $w$ (by $\sup_x |r(x)|/(1-\gamma)$, to be precise).
\end{proof}

In \Cref{fig:example-logistic}, we have extended the logistic map's domain from $[0,1]$ to $\R$ by projecting onto $[0,1]$ in every step. We then added Gaussian noise and obtained the value function. \Cref{thr:differentiability} shows that it must be differentiable.

\section{Conclusion}

We investigated the continuity of the value function in reinforcement learning and optimal control. In particular, we have shown that the value function can be nowhere differentiable, as suggested by \Cref{thr:logistic}, even if the reward function and the underlying system are both relatively well-behaved. Moreover, we have derived an upper bound to the value function's modulus of continuity on bounded state spaces in both discrete and continuous time and have shown sharpness of this bound for linear reward functions. As a consequence, we argued that the value function is Hölder continuous, in line with \cite{bardiOptimalControlViscosity1997,bernsteinAdaptiveresolutionReinforcementLearning2010}, which is useful for bounding variances. Finally, we have shown how to obtain differentiable value functions by introducing ``disturbances''.

\section*{Acknowledgements}
   
H.H. and S.P. acknowledge financial support from the project ``SAIL: SustAInable Life-cycle of Intelligent Socio-Technical Systems'' (Grant ID NW21-059D), which is funded by the program ``Netzwerke 2021'' of the Ministry of Culture and Science of the State of Northrhine Westphalia, Germany.

\appendix

\begin{lemma}
   \label{thr:moc-example}
   Consider the example from \Cref{sec:sharpness_and_example}. Then
   $
   W_v(d_0) = v(1) - v(1-d_0).
   $
\end{lemma}
\begin{proof}
   Note that $v$ is monotonically increasing, which follows from the non-negativity of $v'(x)$. Thus,
   \begin{align*}
      W_v(d_0) = \sup_{d_0 \geq x - x' \geq 0} \{ v(x) - v(x')\}.
   \end{align*}
   Moreover, $v'(x)$ is monotonically increasing, since 
   \begin{align*}
      v''(x) = \sum_{n=0}^\infty (\gamma L)^n (L^n-1) x^{L^n-2} \geq 0.
   \end{align*}
   Hence, for $1 \geq x \geq x' \geq 0$,
   \begin{align*}
      v(x) - v(x') 
      &= \textstyle\int_{x'}^x v'(y) dy \\
      &\leq \textstyle\int_{x'}^x v'(y+(1-x)) dy \\
      &= \textstyle\int_{x'+(1-x)}^1 v'(y) dy \\
      &= v(1) - v(x' + (1 - x)) \\
      &= v(1) - v(1 - (x-x')).
   \end{align*}
   Introducing $\delta = x-x'$, the modulus becomes 
   $$
   W_v(d_0) = v(1) - \inf_{d_0 \geq \delta \geq 0} v(1-\delta).
   $$
   Since $v$ is monotonically increasing, the $\delta$ that minimizes $v(1-\delta)$ should be as large as possible. Since it is bounded by $d_0$, the optimal choice becomes $\delta = d_0$.
\end{proof}

\begin{corollary}[Case $<$]
   \label{thr:continuity-value-function-case-less}
   Suppose \Cref{def:rl-setting-lipschitz} holds with $L < 1/\gamma$ and $R(d_0)=C d_0$. Then \eqref{eq:continuity-value-function-hoelder} is satisfied for $\beta = 1$ and 
   \begin{align*}
      A = C
      \cdot \begin{cases}
         \frac 1 {\ln (1/(\gamma L))} & T = \R_{\geq 0}, \\
         \frac 1 {1-\gamma L} & T = \N_{0}.
      \end{cases}
   \end{align*}
\end{corollary}
\begin{proof}
   Set $d_0 = d(x,x')$. For $T = \N_0$, %
   \begin{align*}
      |v(x)-v(x')| 
      &\leq \sum_{n=0}^\infty \gamma^n C L^n d_0 = \frac{C}{1-\gamma L} d_0.
   \end{align*}
   Similarly, for $ T = \R_{\geq 0}$, 
   \begin{align*}
      |v(x)-v(x')| 
      &\leq \int_{0}^\infty \gamma^t C L^t d_0 dt = \frac{C}{\ln(1/(\gamma L))} d_0.
   \end{align*}
\end{proof}

\begin{corollary}[Case $=$]
   \label{thr:continuity-value-function-case-eq}
   Suppose \Cref{def:rl-setting-lipschitz} holds with $L = 1/\gamma$ and $R(d_0)=C d_0$. For $\beta \in (0,1)$, \eqref{eq:continuity-value-function-hoelder} is satisfied by
   \begin{align*}
      A = C \big( \textstyle\frac{1}{e(1-\beta)} + D^{1-\beta}(\ln D + 1) \big)
      \cdot \begin{cases}
         \frac 1 {\ln (1/\gamma)} & T = \R_{\geq 0}, \\
         \frac 1 {1-\gamma} & T = \N_{0}.
      \end{cases}
   \end{align*}
\end{corollary}
\begin{proof}
      For $T = \R_{\geq 0}$, rewrite $H$ from \Cref{thr:hoelder-integrals} as 
      \begin{equation}
         \label{eq:hc-pf}
         H(d_0) 
         = \frac{C}{\ln(1/\gamma)} d_0 \ln(1/d_0) + \frac{C(\ln D + 1)}{\ln (1/\gamma)} d_0.
      \end{equation}
      In the last term and for any $\beta \in (0,1)$, we can bound $d_0$ by $d_0 = d_0^{1-\beta}d_0^\beta\leq D^{1-\beta} d_0^\beta$. The challenging term is $d_0 \ln(1/d_0)$, for which we would like to find a $B \geq 0$ such that $B d_0^\beta \geq d_0 \ln(1/d_0)$ for every $d_0 \in (0,\infty)$. 
      Divide by $d_0^\beta$ to bound $B$ from below by
      $$
         h(d_0):= d_0^{1-\beta} \ln(1/d_0).
      $$ 
      Notice that
      \begin{align*}
         h'(d_0) = ((1- \beta) \ln (1/d_0) - 1) / d_0^{\beta},
      \end{align*}
      with
      \begin{align*}
         h'(d_0) \geq 0 
         &\iff -\ln d_0 \geq 1/(1-\beta) \\
         &\iff d_0 \leq e^{-1 /(1-\beta)} 
      \end{align*}
      and $h'(d_0^*) = 0 $ for $d_0^* = e^{-1/(1-\beta)}$. Indeed, $h'(d_0) \geq 0$ for $d_0 \in (0,d_0^*]$ and $h'(d_0) \leq 0$ for $d_0 \in [d_0^*, \infty)$. Hence, $d_0^*$ is a global maximum of $h$, and thus 
      \begin{align*}
         B := \frac 1 {e (1-\beta)} = h(d_0^*) \geq h(d_0) \quad \forall d_0 \in (0,\infty).
      \end{align*}
      In \eqref{eq:hc-pf}, we can thus bound $d_0 \ln (1/d_0)$ by $B d_0^\beta$ and $d_0$ by $D^{1-\beta} d_0^\beta$. After factoring out $d_0^\beta$, we obtain 
      $$
         H(d_0) \leq \frac{C}{\ln(1/\gamma)}\big( \frac{1}{e (1-\beta)} + D^{1-\beta}(\ln D + 1) \big) d_0^\beta,
      $$
      which shows the case for $T = \R_{\geq 0}$.
      For $T = \N_0$, use \Cref{thr:hoelder-sums} and then bound $H$ in the same way.
\end{proof}

\begin{corollary}[Case $>$]
   \label{thr:continuity-value-function-case-greater}
   Suppose \Cref{def:rl-setting-lipschitz} holds with $L > 1/\gamma$ and $R(d_0)= C d_0$. Then \eqref{eq:continuity-value-function-hoelder} is satisfied when setting $\beta = \ln(1/\gamma) / \ln L$ and 
   \begin{align*}
      A = \frac{C D^{1-\beta}\ln L}{\ln (\gamma L)} \cdot \begin{cases}
         \frac 1 {\ln (1/\gamma)} & T = \R_{\geq 0}, \\
         \frac 1 { 1-\gamma } & T = \N_{0}.
      \end{cases}
   \end{align*}
\end{corollary}
\begin{proof}
   For $T = \R_{\geq 0}$, apply \Cref{thr:hoelder-integrals} and notice that the term $(\ln \gamma / \ln(\gamma L)) d_0$
   in $K(d_0)$ is non-positive. Relaxing the bound by removing this term gives $A$. %
   Additionally, note that $\beta$ is strictly between 0 and 1 since $L > 1/\gamma > 1$ implies $\ln L > \ln(1/\gamma) > 0$. For $T = \N_0$, we use the same argument but factor in $\ln(1/\gamma)/ (1-\gamma)$ using \Cref{thr:hoelder-sums}.
\end{proof}

\bibliographystyle{plain}
\bibliography{CDC24}

\end{document}